\documentclass[pra,aps,nopacs,onecolumn,twoside,superscriptaddress]{revtex4}

\usepackage{multirow}
\usepackage{amsmath,amsfonts,amssymb,caption,color,epsfig,graphics,graphicx,hyperref,latexsym,mathrsfs,revsymb,theorem,url,verbatim,epstopdf,enumerate}
\usepackage{amsmath,amsfonts,amssymb,caption,hyperref,color,epsfig,graphics,graphicx,latexsym,mathrsfs,revsymb,theorem,url,verbatim,epstopdf,cleveref}
\usepackage{fontenc}
\usepackage{cases}
\hypersetup{colorlinks,linkcolor={blue},citecolor={blue},urlcolor={red}}

\newtheorem{definition}{Definition}
\newtheorem{proposition}[definition]{Proposition}
\newtheorem{lemma}[definition]{Lemma}

\newtheorem{theorem}[definition]{Theorem}
\newtheorem{corollary}[definition]{Corollary}
\newtheorem{conjecture}[definition]{Conjecture}

\newtheorem{remark}[definition]{Remark}
\newtheorem{example}[definition]{Example}
\newtheorem{question}[definition]{Question}
\newtheorem{memo}[definition]{Memo}

\def\squareforqed{\hbox{\rlap{$\sqcap$}$\sqcup$}}
\def\qed{\ifmmode\squareforqed\else{\unskip\nobreak\hfil
		\penalty50\hskip1em\null\nobreak\hfil\squareforqed
		\parfillskip=0pt\finalhyphendemerits=0\endgraf}\fi}
\def\endenv{\ifmmode\;\else{\unskip\nobreak\hfil
		\penalty50\hskip1em\null\nobreak\hfil\;
		\parfillskip=0pt\finalhyphendemerits=0\endgraf}\fi}
\newenvironment{proof}{\noindent \textbf{{Proof.~} }}{\qed}
\def\Dbar{\leavevmode\lower.6ex\hbox to 0pt
	{\hskip-.23ex\accent"16\hss}D}
\makeatletter
\def\url@leostyle{%
	\@ifundefined{selectfont}{\def\UrlFont{\sf}}{\def\UrlFont{\small\ttfamily}}}
\makeatother
\def\bcj{\begin{conjecture}}
	\def\ecj{\end{conjecture}}
\def\bcr{\begin{corollary}}
	\def\ecr{\end{corollary}}
\def\bd{\begin{definition}}
	\def\ed{\end{definition}}
\def\bea{\begin{eqnarray}}
	\def\eea{\end{eqnarray}}
\def\bem{\begin{enumerate}}
	\def\eem{\end{enumerate}}
\def\bex{\begin{example}}
	\def\eex{\end{example}}
\def\bim{\begin{itemize}}
	\def\eim{\end{itemize}}
\def\bl{\begin{lemma}}
	\def\el{\end{lemma}}
\def\bma{\begin{bmatrix}}
	\def\ema{\end{bmatrix}}
\def\bpf{\begin{proof}}
	\def\epf{\end{proof}}
\def\bpp{\begin{proposition}}
	\def\epp{\end{proposition}}
\def\bqu{\begin{question}}
	\def\equ{\end{question}}
\def\br{\begin{remark}}
	\def\er{\end{remark}}
\def\bt{\begin{theorem}}
	\def\et{\end{theorem}}
\def\bmm{\begin{memo}}
	\def\emm{\end{memo}}

\def\btb{\begin{tabular}}
	\def\etb{\end{tabular}}

	\newcommand{\nc}{\newcommand}
	
	\def\a{\alpha}
	\def\b{\beta}

	\def\l{\lambda}

	\def\ps{\psi}

	\def\O{\Omega}
	
	\nc{\bbA}{\mathbb{A}} \nc{\bbB}{\mathbb{B}} \nc{\bbC}{\mathbb{C}}
	\nc{\bbD}{\mathbb{D}} \nc{\bbE}{\mathbb{E}} \nc{\bbF}{\mathbb{F}}
	\nc{\bbG}{\mathbb{G}} \nc{\bbH}{\mathbb{H}} \nc{\bbI}{\mathbb{I}}
	\nc{\bbJ}{\mathbb{J}} \nc{\bbK}{\mathbb{K}} \nc{\bbL}{\mathbb{L}}
	\nc{\bbM}{\mathbb{M}} \nc{\bbN}{\mathbb{N}} \nc{\bbO}{\mathbb{O}}
	\nc{\bbP}{\mathbb{P}} \nc{\bbQ}{\mathbb{Q}} \nc{\bbR}{\mathbb{R}}
	\nc{\bbS}{\mathbb{S}} \nc{\bbT}{\mathbb{T}} \nc{\bbU}{\mathbb{U}}
	\nc{\bbV}{\mathbb{V}} \nc{\bbW}{\mathbb{W}} \nc{\bbX}{\mathbb{X}}
	\nc{\bbZ}{\mathbb{Z}}
	
	\nc{\bA}{{\bf A}} \nc{\bB}{{\bf B}} \nc{\bC}{{\bf C}}
	\nc{\bD}{{\bf D}} \nc{\bE}{{\bf E}} \nc{\bF}{{\bf F}}
	\nc{\bG}{{\bf G}} \nc{\bH}{{\bf H}} \nc{\bI}{{\bf I}}
	\nc{\bJ}{{\bf J}} \nc{\bK}{{\bf K}} \nc{\bL}{{\bf L}}
	\nc{\bM}{{\bf M}} \nc{\bN}{{\bf N}} \nc{\bO}{{\bf O}}
	\nc{\bP}{{\bf P}} \nc{\bQ}{{\bf Q}} \nc{\bR}{{\bf R}}
	\nc{\bS}{{\bf S}} \nc{\bT}{{\bf T}} \nc{\bU}{{\bf U}}
	\nc{\bV}{{\bf V}} \nc{\bW}{{\bf W}} \nc{\bX}{{\bf X}}
	\nc{\bZ}{{\bf Z}}
	
	\nc{\cA}{{\cal A}} \nc{\cB}{{\cal B}} \nc{\cC}{{\cal C}}
	\nc{\cD}{{\cal D}} \nc{\cE}{{\cal E}} \nc{\cF}{{\cal F}}
	\nc{\cG}{{\cal G}} \nc{\cH}{{\cal H}} \nc{\cI}{{\cal I}}
	\nc{\cJ}{{\cal J}} \nc{\cK}{{\cal K}} \nc{\cL}{{\cal L}}
	\nc{\cM}{{\cal M}} \nc{\cN}{{\cal N}} \nc{\cO}{{\cal O}}
	\nc{\cP}{{\cal P}} \nc{\cQ}{{\cal Q}} \nc{\cR}{{\cal R}}
	\nc{\cS}{{\cal S}} \nc{\cT}{{\cal T}} \nc{\cU}{{\cal U}}
	\nc{\cV}{{\cal V}} \nc{\cW}{{\cal W}} \nc{\cX}{{\cal X}}
	\nc{\cZ}{{\cal Z}}
	
	\nc{\hA}{{\hat{A}}} \nc{\hB}{{\hat{B}}} \nc{\hC}{{\hat{C}}}
	\nc{\hD}{{\hat{D}}} \nc{\hE}{{\hat{E}}} \nc{\hF}{{\hat{F}}}
	\nc{\hG}{{\hat{G}}} \nc{\hH}{{\hat{H}}} \nc{\hI}{{\hat{I}}}
	\nc{\hJ}{{\hat{J}}} \nc{\hK}{{\hat{K}}} \nc{\hL}{{\hat{L}}}
	\nc{\hM}{{\hat{M}}} \nc{\hN}{{\hat{N}}} \nc{\hO}{{\hat{O}}}
	\nc{\hP}{{\hat{P}}} \nc{\hR}{{\hat{R}}} \nc{\hS}{{\hat{S}}}
	\nc{\hT}{{\hat{T}}} \nc{\hU}{{\hat{U}}} \nc{\hV}{{\hat{V}}}
	\nc{\hW}{{\hat{W}}} \nc{\hX}{{\hat{X}}} \nc{\hZ}{{\hat{Z}}}
	
	\nc{\hn}{{\hat{n}}}
	
	\def\diag{\mathop{\rm diag}}

	\def\tr{\mathop{\rm Tr}}

	\def\SO{{\mbox{\rm SO}}}
	
	\def\Ort{{\mbox{\rm O}}}
	\def\Un{{\mbox{\rm U}}}
	
	\newcommand{\bra}[1]{\langle#1|}
	\newcommand{\ket}[1]{|#1\rangle}

	\def\Dbar{\leavevmode\lower.6ex\hbox to 0pt
		{\hskip-.23ex\accent"16\hss}D}
\begin{document}
	
\Large

\title{Mutually-orthogonal unitary and orthogonal matrices}

\author{Zhiwei Song}\email[]{zhiweisong@buaa.edu.cn}
\affiliation{School of Mathematical Sciences, Beihang University, Beijing 100191, China}

\author{Lin Chen}\email[]{linchen@buaa.edu.cn (corresponding author)}
\affiliation{School of Mathematical Sciences, Beihang University, Beijing 100191, China}
\affiliation{International Research Institute for Multidisciplinary Science, Beihang University, Beijing 100191, China}
		
\author{Saiqi Liu}\email[]{liu\_saiqi\_sy2324111@buaa.edu.cn (corresponding author)}
\affiliation{School of Mathematical Sciences, Beihang University, Beijing 100191, China}

\begin{abstract}
We introduce the concept of $n$-OU and $n$-OO matrix sets, a collection of $n$ mutually-orthogonal unitary and real orthogonal matrices under Hilbert-Schmidt inner product. We give a detailed characterization of order-three $n$-OO matrix sets under orthogonal equivalence. As an application in quantum information theory, we show that the minimum and maximum numbers of an unextendible maximally entangled bases within a real two-qutrit system are three and four, respectively. Further, we propose a new matrix decomposition approach, defining an $n$-OU (resp. $n$-OO) decomposition for a matrix as a linear combination of $n$ matrices from an $n$-OU (resp. $n$-OO) matrix set. We show that any order-$d$ matrix has a $d$-OU decomposition. As a contrast, we provide criteria for an order-three real matrix to possess an $n$-OO decomposition.
			\end{abstract}
		
		\maketitle

		\section{introduction}
		\label{sec:int}

	 The orthogonality between matrices (or operators) is a fundamental research topic in mathematics and physics. The utilization of an orthogonal basis of matrices to express finite-dimensional matrices is a commonly employed technique in quantum error-correcting codes \cite{1999A}, and the Bloch representation of density operators \cite{Fano1957Description}.	
Some orthogonal bases of the Hilbert-Schmidt spaces have been collected and applied \cite{siewert2022orthogonal}.  
On the other hand,	unitary and real orthogonal matrices play a fundamental role in various areas. 	The Weyl-Heisenberg group implies that there exist $d^2$ order-$d$ mutually-orthogonal unitary operators \cite{weyl1950theory}. As far as we know, the question of determining the number of mutually-orthogonal real orthogonal operators remains open. In this paper, we investigate this problem by  considering a set of $n$ real orthogonal matrices that are mutually-orthogonal. We refer to it as an $n$-OO matrix set. We give a full characterization of order-three $n$-OO matrix sets in Theorems \ref{2ge}-\ref{unt}.

We then apply our results to studying unextendible maximally entangled bases (UMEB) in quantum information theory. This notion refers to a set of incomplete orthonormal maximally entangled states whose complementary space has no maximally entangled states  \cite{bravyi2011unextendible}. It can be used to find quantum channels that are unital but not convex mixtures of unitary operations.
The authors have proved that UMEB do not exist in $\bbC^2 \otimes \bbC^2$, and they constructed a six-number UMEB in $\bbC^3 \otimes \bbC^3$ and twelve-number UMEB in $\bbC^4 \otimes \bbC^4$ \cite{bravyi2011unextendible}.  Later, UMEB in arbitrary bipartite spaces has been investigated \cite{chen2013unextendible,shi2019constructions, li2014unextendible,wang2014unextendible,guo2016constructing}.  For example, there exists a $d^2$-member UMEB in $\bbC^d\otimes \bbC^{d'}$ ($\frac{d'}{2}<d<d'$)\cite{chen2013unextendible}.
However, characterizing UMEB in real Hilbert space remains a problem. It is also an interesting question to give the minimum and
maximum numbers of UMEB as a function of
dimension.	In this article, we present the detailed forms of UMEB in two-qutrit real Hilbert space in Theorem \ref{ru}. We show that its minimum and maximum numbers are three and four, respectively.

Another interesting topic is the additive decomposition of matrices in terms of unitary and real orthogonal matrices \cite{olsen1986convex,wu1994additive,zhan2001span,li2002additive,bottcher2012orthogonal}. In quantum physics, such a decomposition is related to the  mixed-unitary channels \cite{girard2022mixed}, the Kraus operators describing an open system \cite{cui2012optimal,suri2023two}. It is well known that every matrix can be written as the average of two unitary matrices \cite{wu1994additive,zhan2001span}. For the real case, every  matrix can be written as the linear combination of no more than four orthogonal matrices \cite{li2002additive}. Research has also explored the singular value structure of a matrix under such decomposition \cite{li2022singular}.		
However,  there appears to be no prior research whether a matrix can be written as the linear combination of $n$ mutually-orthogonal unitary or real orthogonal matrices. In this paper, we introduce and define such a decomposition as $n$-OU and $n$-OO decomposition, respectively. We show that all complex matrices possess an $n$-OU decomposition in Theorem \ref{le:nxnUnitary}. In contrast, we provide conditions for an order-three real matrix to have an $n$-OO decomposition, as outlined in Theorem \ref{condition}. 

	The rest of this paper is organized as follows. In Sec. \ref{sec:pre} we introduce the basic facts for this paper. In Sec. \ref{sec:nOO+nOU} we introduce the $n$-OO matrix and $n$-OU matrix sets. Then we introduce the $n$-OU and $n$-OO decomposition in Sec. \ref{sec:ooude}. We further apply our results in quantum information theory in Sec. \ref{sec:app}. Finally we conclude in Sec. \ref{sec:pro}.

	\section{preliminaries}
\label{sec:pre}
In this section, we introduce the main notions and theories used in this paper. 
We denote $M^\dagger$ as the 
conjugate transpose of $M$, and $I_n$ as the order-$n$ identity matrix. Next we respectively denote $\Un(n)$ and $\Ort(n)$ as the unitary and real orthogonal groups of order-$n$ matrices. The matrices in $\Ort(n)$ with 
determinant 1 form the special orthogonal group $\SO(n)$. 
Viewing matrices as elements in a space of linear bounded operators endowed with Hilbert-Schmidt inner product,
we refer that two complex (resp. real) matrices $M$ and $N$ are orthogonal if $\tr(M^\dagger N)=0$ (resp. $\tr(M^T N)=0$). Now we are ready to present the main definitions of this section.

\begin{definition}
	\label{df}
	(i)	Suppose a matrix set has $n$ elements $U_1,\cdots,U_n$, which are all unitary matrices. They are mutually orthogonal, i.e., $\tr(U_i^\dagger U_j)=0$ for any $i\neq j$. We refer to the set as an $n$-OU matrix set.
	
	(ii) Suppose a matrix set has $n$ elements $O_1,\cdots,O_n$, which are all real orthogonal matrices. They are mutually orthogonal, i.e., $\tr(O_i^TO_j)=0$ for any $i\neq j$. We refer to the set as an $n$-OO matrix set.
\end{definition}

It is clear that if $\{M_1,\cdots,M_n\}$ is an $n$-OU (resp. $n$-OO) matrix set, then $\{(-1)^{a_i}M_1,\cdots,(-1)^{a_n}M_n\}$ is also an $n$-OU (resp. $n$-OO) matrix set where the integers $a_1,\cdots,a_n$  take $0$ or $1$. For convenience, we shall omit $(-1)^{a_i}$ when referring to the $n$-OU and $n$-OO matrix set in the following. To characterize the relation between matrix sets, we construct the following definition.  

\begin{definition}
	Suppose 
	$\mathcal{M}:=\{M_1,\cdots,M_n\}$ and $\mathcal{N}:=\{N_1,\cdots,N_n\}$ are two $n$-OU matrix (reps. $n$-OO) matrix sets. Further, there exist two unitary  (reps. real orthogonal) matrices $U$ and $V$ such that $UM_iV=(-1)^{a_i}N_{\pi(i)}$ for $i=1,\cdots,n$, where $a_i\in \{0,1\}$, $\pi(i)$ is a permutation of $1,\cdots,n$. Then we say that $\mathcal{M}$ and $\mathcal{N}$
	are unitarily (reps. orthogonally) equivalent.
\end{definition}

There has been some research in the unitary and orthogonal equivalence of two sets of general matrices \cite{gerasimova2013simultaneous,jing2015unitary}. It has been shown that the orthogonal equivalence of two sets of real matrices has applications in determining the local unitary equivalence between two quantum states \cite{gour2013classification,jing2014slocc}.
The following lemma gives a necessary condition of two orthogonally equivalent $n$-OO matrix sets, which can be proved straightforwardly.
\begin{lemma}
	\label{s}
	Suppose two $n$-OO matrix sets $\{M_1,\cdots,M_n\}$ and $\{N_1,\cdots,N_n\}$ are orthogonally equivalent, i.e., $UM_iV=(-1)^{a_i}N_{\pi(i)}$ for $i=1,\cdots,n$. Then for any real numbers $k_1,\cdots,k_n$,
	the two matrices $\sum_{i=1}^n k_i M_i$ and $\sum_{i=1}^n k_i (-1)^{a_i}N_{\pi(i)}$ have the same singular values.
\end{lemma}

We now define a new kind of matrix decomposition in terms of $n$-OO and $n$-OU matrix sets. 

\begin{definition}
	\label{defd}
	(i)	If a matrix $M$ can be written as the linear combination of $n$ matrices which form an $n$-OU matrix set, then we say that $M$ has an $n$-OU decomposition.   
	
	(ii) 	If a real matrix $M$ can be written as the real linear combination of $n$ matrices which form an $n$-OO matrix set, then we say that $M$ has an $n$-OO decomposition.   
\end{definition}

Now we are ready to present the main results in the following sections.

	\section{ $n$-OO matrix and $n$-OU matrix sets}
		\label{sec:nOO+nOU}
	This section considers the matrix sets in Definition \ref{df}.	In Subsection \ref{ooA}, we give a characterization of order-three $n$-OO matrix sets. In Subsection \ref{ooB}, we list some known results of $n$-OU matrix sets and $n$-OO matrix sets beyond order three.
		\subsection{order-three $n$-OO matrix sets}
		\label{ooA}
		\begin{lemma}
			\label{so2}
			For any matrix $O\in \Ort(3)$, there exist 
			$O_\a, O_\beta\in \SO(2)$ and $R\in \Ort(2)$ such that $(1\oplus O_{\alpha})O(1\oplus O_{\beta})=R\oplus 1$.
			\end{lemma}
		\begin{proof}
			Any element $O_{\alpha}$ in $\SO(2)$ can be written as $\bma 
			\cos \alpha & -\sin \alpha \\
			\sin \alpha & \cos \alpha 
			\ema$. We can choose a real number $\alpha$ such that the entry in the third row and the first column of $(1\oplus O_{\alpha})O$ is zero. Next, we can choose a real number $\beta$ such that the first two entries in the third row of $(1\oplus O_a)O(1\oplus O_\b)$ are both zero. Thus the third diagonal entry of $(1\oplus O_a)O(1\oplus O_\b)$ is 1 or -1, and if it is $-1$ then we left multiply ($1\oplus -I_2$) and obtain the result.
		\end{proof}
		
The next theorem and  corollary give a characterization of order-three $n$-OO matrix sets for $n\le 3$. 
	\begin{theorem}
		\label{2ge}
		Let $\O_1=\bma 1&0&0\\0&-\frac{1}{2}&-\frac{\sqrt{3}}{2}\\0&\frac{\sqrt{3}}{2}&-\frac{1}{2}\ema$ and $\Omega_2=\bma -\frac{1}{3}& \frac{\sqrt{2}}{3}&\frac{\sqrt{6}}{3}\\
		\frac{\sqrt{2}}{3}&\frac{5}{6}&-\frac{\sqrt{3}}{6}\\-\frac{\sqrt{6}}{3}&\frac{\sqrt{3}}{6}&-\frac{1}{2}\ema$.
		
		(i) Any order-three 2-OO matrix set is orthogonally equivalent to $\{I_3, \O_1\}$.
		
		(ii) Any order-three 3-OO matrix set is orthogonally equivalent to $\mathcal{C}$ or $\mathcal{D}$, where \begin{eqnarray}
			\label{c1}
			&&\mathcal{C}:=\{I_3,\O_1,\O_1^T\},\\
			\label{c2}
			&&\mathcal{D}:=\{I_3,\O_1,\O_2\},
		\end{eqnarray}
		 are not orthogonally equivalent.
			\end{theorem}
	\begin{proof}
	(i) Suppose $\{O_1,O_2\}$ is an order-three 2-OO matrix set. We have $\tr(O_1^TO_2)=0$. Note that the eigenvalues of $O_1^TO_2\in \Ort(3)$ have modules one and the complex eigenvalues appear in conjugate pair. The traceless condition implies that the eigenvalues of $O_1^TO_2$ are either $1,-\frac{1}{2}+\frac{\sqrt{3}}{2}i,-\frac{1}{2}-\frac{\sqrt{3}}{2}i$ or $-1,\frac{1}{2}+\frac{\sqrt{3}}{2}i,\frac{1}{2}-\frac{\sqrt{3}}{2}i$.
   Based on the quasi-diagonalization of real matrix (see \cite[Theorem 2.5.8, p136]{horn2012matrix}), there exists  $P\in \Ort(3)$ such that $O_1^TO_2=(-1)^aP^T\O_1 P$ where $a=0$ or 1 corresponds to the two cases of the eigenvalues. Further, we have
   \begin{eqnarray}
   	(PO_1^T)O_1P^T=I_3, \quad (PO_1^T)O_2P^T=(-1)^a\O_1.
   \end{eqnarray}
  This implies that  $\{O_1,O_2\}$ is orthogonally equivalent to $\{I_3, \O_1\}$.

(ii) One can verify that both $\mathcal{C}$ and $\mathcal{D}$ are 3-OO matrix sets. We first show that $\mathcal{C}$ and $\mathcal{D}$ are not orthogonally equivalent. Otherwise by Lemma \ref{s}, there exist $a_1,a_2,a_3\in \{0,1\}$ such that the matrices $S_1:=I_3+\O_1+\O_1^T$ and $S_2:=(-1)^{a_1}I_3+(-1)^{a_2}\O_1+(-1)^{a_3}\O_2$ have the same singular values.  There are eight cases in terms of $a_1,a_2,a_3$. For each case, by calculation, we obtain that the singular values of $S_2$ are not the same as that of $S_1$.  
This is a contradiction. Hence $\mathcal{C}$ and $\mathcal{D}$ 
are not orthogonally equivalent.

We next prove that any order-three 3-OO matrix set $\{O_1,O_2,O_3\}$ is orthogonally equivalent to $\mathcal{C}$ or $\mathcal{D}$. From (i), we know that $\{O_1,O_2\}$ is orthogonally equivalent to $\{I_3,\O_1\}$. Thus we only need to find $O_3'\in \Ort(3)$ satisfying $\tr(O_3')=0$ and $\tr(\O_1^TO_3')=0$. The first condition implies that  there exists $Q\in \Ort(3)$ such that $O_3'=(-1)^{a}Q^T\O_1Q$ where $a\in \{0,1\}$. Using Lemma \ref{so2}, we can write 
$Q=(1\oplus O_\alpha)(R\oplus 1)(1\oplus O_\beta) $ where $O_\alpha,O_\beta\in \SO(2)$ and $R\in \Ort(2)$. Note that $\SO(2)$ is an abelian group, we have 
\begin{eqnarray}
	O_3'=(-1)^{a}(1\oplus O_\beta^T)(R^T\oplus 1)\O_1 (R\oplus 1)(1\oplus O_\beta).
\end{eqnarray}
Next, let 
\begin{eqnarray}
	\label{o3''}
	O_3''=(-1)^{a}(R^T\oplus 1)\O_1 (R\oplus 1).
\end{eqnarray}
One can verify that $\{I_3,\O_1,O_3'\}$ is orthogonally equivalent to $\{I_3,\O_1,O_3''\}$ through the matrix $1\oplus O_\beta$ and its transpose.
The orthogonality between $O_3''$ and $\O_1$ implies that the matrix $R$ in (\ref{o3''}) satisfies
\begin{eqnarray}
	\label{eq:trbma}	
	\tr(\O_1^T (R^T\oplus1)\O_1 (R\oplus1))=0.
\end{eqnarray}
We use the Software Mathematica to solve the equation (\ref{eq:trbma}). There are six solutions:
\begin{eqnarray}
	\notag
	&&R_1=\bma -1&0\\0&-1\ema,R_2=\bma 1&0\\0&-1\ema,
	R_3=\bma \frac{1}{3}&\frac{2\sqrt{2}}{3}\\-\frac{2\sqrt{2}}{3}&\frac{1}{3}\ema,
	R_4=\bma -\frac{1}{3}&-\frac{2\sqrt{2}}{3}\\-\frac{2\sqrt{2}}{3}&\frac{1}{3}\ema,\\
	&&R_5=\bma \frac{1}{3}&-\frac{2\sqrt{2}}{3}\\\frac{2\sqrt{2}}{3}&\frac{1}{3}\ema,
	R_6=\bma -\frac{1}{3}&\frac{2\sqrt{2}}{3}\\\frac{2\sqrt{2}}{3}&\frac{1}{3}\ema.
\end{eqnarray}
Further calculations imply that 
\begin{eqnarray}
	\notag
	&&(R_1^T\oplus1)\O_1 (R_1\oplus1)
	=(R_2^T\oplus1)\O_1 (R_2\oplus1)=\O_1^T,\\
	\notag
	&&(R_3^T\oplus1)\O_1 (R_3\oplus1)
	=(R_4^T\oplus1)\O_1 (R_4\oplus1)=\O_2,\\
	&&(R_5^T\oplus1)\O_1 (R_5\oplus1)=(R_6^T\oplus1)\O_1 (R_6\oplus1).
\end{eqnarray}	
Note that $R_5=R_3^T$. This implies that $\{I_3,\O_1,(R_3^T\oplus1)\O_1 (R_3\oplus1)\}$ and $\{I_3,\O_1,(R_5^T\oplus1)\O_1 (R_5\oplus1)\}$ are orthogonally equivalent. Hence the six solutions correspond to the two 3-OO matrix sets $\mathcal{C}$ in (\ref{c1}) and $\mathcal{D}$ in (\ref{c2}). 
\end{proof}

		\begin{corollary}
		\label{3onew}
	Let
		\begin{eqnarray}
			\label{G14}
	G_1=\bma 1&0&0\\0&0&1\\0&1&0\ema, G_2=\bma 0&0&1\\0&1&0\\1&0&0\ema, G_3=\bma 0&-1&0\\-1&0&0\\0&0&1\ema, G_3'=\bma 0&1&0\\1&0&0\\0&0&1\ema
		\end{eqnarray}
	and 
		\begin{eqnarray}
			\label{c1'}
			&&\mathcal{G}_1:=\{G_1,G_2,G_3'\},\\
			\label{c2'}
			&&\mathcal{G}_2:=\{G_1,G_2,G_3\}.
		\end{eqnarray}
		Then the 3-OO matrix set $\mathcal{C}$ in (\ref{c1}) is orthogonally equivalent to $\mathcal{G}_1$, $\mathcal{D}$ in (\ref{c2}) is orthogonally equivalent to $\mathcal{G}_2$. As a consequence, any order-three 3-OO matrix set is orthogonally equivalent to $\mathcal{G}_1$ or $\mathcal{G}_2$.
	\end{corollary}
	\begin{proof}
			Let $B=\bma \frac{1}{\sqrt{3}} & -\sqrt{\frac{2}{3}} & 0 \\
		\frac{1}{\sqrt{3}} & \frac{1}{\sqrt{6}} & -\frac{1}{\sqrt{2}} \\
		\frac{1}{\sqrt{3}} & \frac{1}{\sqrt{6}} & \frac{1}{\sqrt{2}}\ema \in \Ort(3)$. The claim is directly proven by the following identities,
\begin{alignat}{2}
	\notag
	(B^TG_1)G_1B &= I_3, \quad &(B^TG_1)G_2B &= \O_1,\\
	(B^TG_1)G_3'B &= \O_1^T, \quad &(B^TG_1)G_3B &= \O_2.
\end{alignat}
	\end{proof}
	
	Using the above results, we continue to characterize order-three $n$-OO matrix sets for $n>3$.
		\begin{theorem}
			\label{unt}
	(i)		There do not exist any $O\in \Ort(3)$ that is orthogonal to the three matrices of $\mathcal{G}_1$ in (\ref{c1'}). In other words, the 3-OO matrix set $\mathcal{G}_1$ is unextendible.
	
	(ii) Let $G_4=\bma-\frac{1}{2}&\frac{-1-\sqrt{5}}{4}&\frac{1-\sqrt{5}}{4}\\\frac{-1+\sqrt{5}}{4}&-\frac{1}{2}&\frac{1+\sqrt{5}}{4}\\\frac{1+\sqrt{5}}{4}&\frac{1-\sqrt{5}}{4}&-\frac{1}{2}\ema$.	Any order-three 4-OO matrix set is orthogonally equivalent to
	\begin{eqnarray}
		\label{4o}
		\mathcal{G}:=\{G_1,G_2,G_3,G_4\}.
	\end{eqnarray}
	
	(iii) There does not exist any order-three $n$-OO matrix set for $n\ge 5$. 
		\end{theorem}
		\begin{proof}
	(i) We use the Software Mathematica to find such a matrix  based on the orthogonality conditions. There does not exist any solution. Hence the claim is proven.
	
	(ii) Suppose $\{O_1,O_2,O_3,O_4\}$ is an order-three 4-OO matrix set. Using Corollary \ref{3onew} and (i), we conclude that $\{O_1,O_2,O_3\}$ can only be orthogonally equivalent to $\mathcal{G}_2$. Consequently, we only need to find all the matrices that are orthogonal to the three matrices $G_1,G_2,G_3$ in (\ref{G14}). By calculation, the solutions are $\pm G_4$ and $\pm G_4'$ where
	$G_4'= \bma-\frac{1}{2}&\frac{-1+\sqrt{5}}{4}&\frac{1+\sqrt{5}}{4}\\\frac{-1-\sqrt{5}}{4}&-\frac{1}{2}&\frac{1-\sqrt{5}}{4}\\\frac{1-\sqrt{5}}{4}&\frac{1+\sqrt{5}}{4}&-\frac{1}{2}\ema$. Further, one can verify that
	\begin{eqnarray}
		\notag
	&&G_3'G_1G_3'=G_2, \quad G_3'G_2G_3'=G_1\\
	&&G_3'G_3G_3'=G_3, \quad G_3'G_4G_3'=G_4'.
	\end{eqnarray}
This implies that the two sets $\{G_1,G_2,G_3,G_4\}$ and $\{G_1,G_2,G_3,G_4'\}$ are orthogonally equivalent. We have proved the assertion.

(iii) Suppose there exists an order-three 5-OO matrix set. From (ii), we know that four of them are orthogonally equivalent to the 4-OO matrix set $\mathcal{G}$. So we only need to find an orthogonal matrix that is orthogonal to $G_1,G_2,G_3,G_4$. However, applying the calculations by the Software Mathematica, we do not find any solutions. Hence we conclude that the claim is proven.
		\end{proof}

			\subsection{$n$-OU and $n$-OO matrix sets of other orders}
		\label{ooB}
		\begin{proposition}
			\label{ou}
			For any order-$d$, there exists a $d^2$-OU matrix set. The element in the set is
			\begin{eqnarray}
				U_{n,m}=\sum_{k=0}^{d-1}\omega_d^{kn}\ket{k\oplus m}\bra{k},
			\end{eqnarray}
			where $n,m=0,1,\cdots,d-1$, $\omega_d$ is any primitive $d$th root of unity, and $k\oplus m$ denotes $k+m$ mod $d$. These matrices form a group which corresponds to the Weyl–Heisenberg group \cite{weyl1950theory}.
		\end{proposition}
		
		\begin{proposition}
			\label{2k}
			For any order $d=2^k$ where $k$ is a positive integer, there exists a $d^2$-OO matrix set $\{A_{i_1}\otimes \cdots \otimes A_{i_k}\}_{i_1,\cdots,i_k=0}^3$, where
			\begin{eqnarray}
				\label{2jie}
				A_0=\bma 1&0\\0&1\ema, A_1=\bma 1&0\\0&-1\ema, A_2=\bma 0&1\\1&0\ema, A_3=\bma 0&1\\-1&0\ema.
			\end{eqnarray}
		\end{proposition}
	
		From above, we know that $d^2$ is the largest number of $n$  for an $n$-OU matrix set,  and an $n$-OO matrix set if $d=2^k$.  It is then natural to ask for order $d\neq 2^k$, what is the largest number of $n$ for an $n$-OO matrix set. We have shown that the number is four when $d=3$. The following shows that it is not less than $d$, but how to show the upper bound remains an open problem.
		\begin{proposition}
			\label{noo}
			There exists an order-$d$ $d$-OO matrix set for any positive integer $d$. The element $P_k$ ($k=0,1,\cdots,d-1$) is the signed permutation matrix whose $(i,i+k \mod d)$ entries are $\pm 1$ and other entries are 0.
		\end{proposition}
		
		\section{$n$-OU  and $n$-OO matrix decomposition}
		\label{sec:ooude}
	In this section, we analyze the $n$-OU and $n$-OO matrix decomposition in subsection \ref{sec:res} and \ref{3ode}, respectively. In subsection \ref{subsec:weak}, we present two weaker forms of $n$-OO decomposition for order-three real matrices. Before proceeding, we propose the following preliminary lemma.
	\begin{lemma}
		\label{tran}
If a complex (resp. real) order-$d$ matrix $M$ has an $n$-OU (resp. $n$-OO) decomposition, then any order-$d$ matrix $N$, which has the same singular values as that of $M$, also has an $n$-OU (resp. $n$-OO) decomposition.
	\end{lemma}
	\begin{proof}
Suppose the matrix $N$ has the same singular values as $M$, which are $\sigma_1,\cdots,\sigma_d$. Using SVD, there exist $V_1,W_1,V_2,W_2\in \Un(d)$ such that 
	\begin{eqnarray}
V_1NW_1=\diag(\sigma_1,\cdots,\sigma_d)=V_2MW_2.
\end{eqnarray}
 If $M$ has an $n$-OU decomposition in terms of the $n$-OU matrix set $\{U_1,\cdots, U_n\}$, the above identity implies that $N$ has an $n$-OU decomposition in terms of the $n$-OU matrix set $\{V_1^\dagger V_2U_1W_1W_2^\dagger,\cdots, V_1^\dagger V_2U_nW_1W_2^\dagger\}$. The argument holds similarly for the $n$-OO decomposition of real matrices.
	\end{proof}

	\subsection{The $n$-OU decomposition of complex matrices}	
	\label{sec:res}
		\begin{theorem}
			\label{le:nxnUnitary}
			(i)	Every order-$d$ matrix has an $d$-OU decomposition.
			
			(ii) There exist order-$d$ matrices that do not have any $(d-1)$-OU decomposition.
		\end{theorem}
		\begin{proof}
			(i)	Using SVD and Lemma \ref{tran}, we only need to show that any diagonal positive semidefinite matrix has a $d$-OU decomposition.
			 Consider the order-$d$ complex matrix $K=\bma 1&1&1&\cdots&1\\1&z&z^2&\cdots&z^{d-1}\\1&z^2&z^4&\cdots&z^{2(d-1)}\\ \vdots&\vdots&\vdots&\ddots&\vdots\\1&z^{d-1}&z^{2(d-1)}&\cdots&z^{(d-1)^2}\ema$, where $z^d=1$ and $z\ne1$. One can verify that each entry of $K$ has modulus one and the column vectors of $K$ are mutually orthogonal.  
			Let $Z_i$ be the order-$d$ diagonal matrix whose diagonal entries correspond uniquely to the elements of the  $i$-th column vector of $K$. We obtain that $\{Z_1,\cdots,Z_d\}$ is a $d$-OU matrix set. Consequently, any order-$d$ diagonal matrix is in the linear spanning of $Z_1,\cdots,Z_d$. Thus the claim is proven.

			(ii) We construct such an example of order-$d$ matrix. Suppose  
			\begin{eqnarray}
				\label{eq:A=k1O1}
				\label{A}
				\diag(1,0,\cdots,0)=k_1U_1+k_2U_2\cdots+k_{d-1}U_{d-1},
			\end{eqnarray}
			where $k_1,\cdots,k_{d-1}$ are complex numbers and $\{U_1,\cdots,U_{d-1}\}$ is a $(d-1)$-OU matrix set. Denote $u_{i}$ as the first diagonal entry of $U_i$. By left-multiplying the matrix $U_i^\dagger$ on both sides of Eq. \eqref{A} and taking trace, we obtain that
			\begin{eqnarray}
			k_i=\frac{1}{d}\bar{u_{i}}
			\end{eqnarray}
	holds for any $1\le i\le d-1$. Taking $k_i$ to Eq.\eqref{A} and comparing the first diagonal entry of two sides, we have
			\begin{eqnarray}
				1=k_1u_{1}+\cdots+k_{d-1}u_{d-1}=\frac{1}{d}|u_{1}|^2+\cdots+\frac{1}{d}|u_{d-1}|^2.
			\end{eqnarray}
			This is impossible because $|u_{1}|^2,\cdots,|u_{d-1}|^2\le 1$ due to that $U_i$ is unitary. So the decomposition in \eqref{eq:A=k1O1}
			does not exist. 
		\end{proof}

		\subsection{The $n$-OO decomposition of order-three real matrices}
		\label{3ode}

	In this subsection,	we consider the $n$-OO decomposition for order-three real matrices. Different from Theorem \ref{le:nxnUnitary}, we propose the following result.
		
		\begin{theorem}
			\label{condition}
Suppose $M$ is an order-three real matrix whose singular values are $x,y,z$. 
		
			(i) If two of $x,y,z$ are equal, w.l.o.g., $x=z$, further, $2x\ge y$, then $M$ has a 2-OO decomposition.
			
			(ii)  The matrix $M$ has a 3-OO decomposition if and only if $x,y,z$ satisfy one of the following two conditions,
			
		\quad	(a) at least two of $x,y,z$ are equal;
			
		\quad	(b) There exists a real solution for the following equations in terms of $l_1,l_2,l_3$,
			\begin{eqnarray}
				\label{l1}
				-l_1-l_2-l_3&=&-x'-y'-z',\\\label{l2}
				l_1l_2+l_1l_3+l_2l_3-l_1^2 - l_2^2- l_3^2&=&x'y'+x'z'+y'z',\\
				\label{l3}
				l_1^3+l_2^3+l_3^3+l_1l_2l_3&=&-x'y'z',
			\end{eqnarray}
			where $x'\in\{-x,x\}$, $y'\in\{-y,y\}$ and $z'\in\{-z,z\}$.

			(iii) The matrix $M$ has a 4-OO decomposition if and only if there exists a real solution for the following equations in terms of $l_1,l_2,l_3,l_4$,
			\begin{eqnarray}
				\label{l5}
				&&\alpha=-(x^2+y^2+z^2),\\
				\label{l6}
				&&\beta=x^2y^2+x^2z^2+y^2z^2,\\
				\label{l7}
				&&\gamma=-x^2y^2z^2,
			\end{eqnarray}
			where 
			\begin{eqnarray}
							\notag
							\label{alpha}
							\alpha&&=-3\sum_{i=1}^4 l_i^2,\\
							\notag
							\beta&&=3\sum_{i,j=1}^4 l_i^2l_j^2 +
	2l_1^2(l_2l_3+l_2l_4+l_3l_4)+2l_2^2(l_1l_3+l_1l_4+l_3l_4)\\
		                    \notag
	&&+2l_3^2(l_1l_2+l_1l_4+l_2l_4)+2l_4^2(l_1l_2+l_1l_3+l_2l_3)
							-6l_1l_2l_3l_4, \\
							\label{gamma}
							\gamma&&=-(\sum_{i=1}^4 l_i^3 + l_1 l_3 l_4 + l_2 l_3 l_4  + l_1 l_2 l_3 + l_1 l_2 l_4)^2.
						\end{eqnarray}
	\end{theorem}

		\begin{proof}
			(i) Using SVD and Lemma \ref{tran}, the claim is proved by the following identity:
			\begin{eqnarray}
				\label{2o}
				\notag
				\bma x&0&0\\0&x&0\\0&0&y \ema
				&&=\frac{1}{6} \left(3y + \sqrt{12x^2 - 3y^2}\right)\begin{bmatrix}
					\frac{y + \sqrt{12x^2 - 3y^2}}{4x} & \frac{-\sqrt{3}y + \sqrt{4x^2 - y^2}}{4x} & 0 \\
					\frac{\sqrt{3}y - \sqrt{4x^2 - y^2}}{4x} & \frac{y + \sqrt{12x^2 - 3y^2}}{4x} & 0 \\
					0 & 0 & 1
				\end{bmatrix}\\
				&&+\frac{1}{6} \left(-3y + \sqrt{12x^2 - 3y^2}\right)
				\begin{bmatrix}
					\frac{-y + \sqrt{12x^2 - 3y^2}}{4x} & -\frac{\sqrt{3}y + \sqrt{4x^2 - y^2}}{4x} & 0 \\
					\frac{\sqrt{3}y + \sqrt{4x^2 - y^2}}{4x} & \frac{-y + \sqrt{12x^2 - 3y^2}}{4x} & 0 \\
					0 & 0 & -1
				\end{bmatrix},
			\end{eqnarray}
	where the two matrices on the rightside consist a 2-OO matrix set.
			
		(ii) Suppose $M$ has a 3-OO decomposition. Using Lemma \ref{s} and Corollary \ref{3onew},
		there exist $U,V\in \Ort(3)$ and real numbers $l_1,l_2,l_3$ such that $UMV=l_1G_1+l_2G_2+l_3G_3$ or $UMV=l_1G_1+l_2G_2+l_3G_3'$ where $G_1,G_2,G_3,G_3'$ are defined in Corollary \ref{3onew}.
		This implies that the singular values of either $L_a:=\bma l_1&l_3&l_2\\l_3&l_2&l_1\\l_2&l_1&l_3\ema$ or $L_b:=\bma l_1&-l_3&l_2\\-l_3&l_2&l_1\\l_2&l_1&l_3\ema$ are $x,y,z$. Note that $L_a$ and $L_b$ are real symmetric matrices, the singular values of $L_a$ (resp. $L_b$) are the absolute values of the  eigenvalues of $L_a$ (resp. $L_b$). By a direction calculation, the eigenvalue set of $L_a$ is 
		\begin{eqnarray}
			\label{det}
			\notag
		\{l_1+l_2+l_3,\sqrt{\frac{1}{2}(l_1-l_2)^2+\frac{1}{2}(l_1-l_3)^2+\frac{1}{2}(l_2-l_3)^2},\\-\sqrt{\frac{1}{2}(l_1-l_2)^2+\frac{1}{2}(l_1-l_3)^2+\frac{1}{2}(l_2-l_3)^2}\}.
		\end{eqnarray}
	Hence if  $L_a$ has singular values $x,y,z$, then at least two of $x,y,z$ are equal. On the other hand, suppose $L_b$ has singular values $x,y,z$. The eigenfunction of $L_b$ is
	\begin{eqnarray}
		\label{det}
		\notag
		&&\det(\lambda I_3-L_b)
		\\ \notag
		=&&\lambda^3-(l_1+l_2+l_3)\lambda^2+(l_1l_2+l_1l_3+l_2l_3-l_1^2 - l_2^2- l_3^2 )\lambda+l_1^3+l_2^3+l_3^3+l_1l_2l_3,\\
	\end{eqnarray}
where the three zeroes are  $x'\in \{x,-x\}$, $y'\in \{y,-y\}$ and $z'\in \{z,-z\}$. According to the relationship between roots and coefficients of (\ref{det}), we obtain that $\{x,y,z\}$ must satisfy the condition (b).
	
	Conversely, suppose $x,y,z$ satisfy condition (a) or (b). Then there exist real numbers $l_1,l_2,l_3$ such that the singular values of $L_a$ or $L_b$ are $x,y,z$. Since both $L_a$ and $L_b$ have a 3-OO decomposition, using Lemma \ref{tran}, we obtain that $M$ also has a 3-OO decomposition.
	
		(iii) The proof is similar to (ii). Suppose  $M$ has a 4-OO decomposition. By Theorem \ref{unt} (ii), there exist $U,V\in \Ort(3)$ and real numbers $l_1,l_2,l_3,l_4$ such that $UAV=l_1G_1+l_2G_2+l_3G_3+l_4G_4$ where $G_1,G_2,G_3,G_4$ are defined in Corollary \ref{3onew} and Theorem \ref{unt}. Hence the singular values of $S:=l_1G_1+l_2G_2+l_3G_3+l_4G_4$ are $x,y,z$ and thus the eigenvalues of $S^TS$ are $x^2,y^2,z^2$. By calculation, the  eigenfunction of $S^TS$ is:
	\begin{eqnarray}
		\label{det1}
		\det(\lambda I-S^TS)
		=\lambda^3+\alpha\lambda^2+\beta\lambda+\gamma,
	\end{eqnarray}
	where $\alpha,\beta,\gamma$ are in (\ref{alpha}).
   According to the relationship between roots and coefficients of (\ref{det1}), we obtain that the equations (\ref{l5})-(\ref{l7}) must have a real solution.

Conversely, suppose there exists a real solution of (\ref{l5})-(\ref{l7}). Then there exist real numbers $l_1,l_2,l_3,l_4$ such that the singular values of $S$ are $x,y,z$. Since $S$ has a 4-OO decomposition in terms of $G_1,G_2,G_3,G_4$, using Lemma \ref{tran}, we obtain that $M$ has a 4-OO decomposition.
\end{proof}

		We next propose some examples to illustrate the above theorem.
		
			{\bf Example 1:}
		Suppose at least two of the singular values of $M$ are equal. Firstly, 
	use SVD to obtain $U,V\in \Ort(3)$ such that $UMV=\diag(y,x,x)$.
	The 3-OO decomposition of $M$ is 	
	\begin{eqnarray}
	M=U\bma y&0&0\\0&x&0\\0&0&x \ema V=	U(\frac{y+2x}{3}I_3+\frac{y-x}{3}\O_1+\frac{y-x}{3}\O_1^T)V,
	\end{eqnarray}	
	where $\{I_3,\O_1,\O_1^T\}$ is the 3-OO matrix set defined in (\ref{c1}).

		{\bf Example 2:}
			Suppose the singular values of $M$ are $1,3,3-\sqrt{2}$. For the equations (\ref{l1})-(\ref{l3}), if we take $x'=1,y'=3,z'=-(3-\sqrt{2})$, then there is a real solution: $l_1=l_2=\sqrt{2},l_3=1-\sqrt{2}$.
		Hence $M$ has a 3-OO decomposition.
		To obtain this decomposition, we first use SVD  to find $U_1,V_1,U_2,V_2\in \Ort(3)$ such that $U_1MV_1=\diag(1,3,3-\sqrt{2})=U_2\left(\sqrt{2}G_1+\sqrt{2}G_2+(1-\sqrt{2})G_3\right)V_2$, where $G_1,G_2,G_3$ are defined in Corollary \ref{3onew}. Then $M=U_1^\dagger U_2\left(\sqrt{2}G_1+\sqrt{2}G_2+(1-\sqrt{2})G_3\right)V_2V_1^\dagger$.
	
		{\bf Example 3:}
		Suppose the singular values of $M$ are $0,1,\sqrt{6}$. By calculation, the equations (\ref{l1})-(\ref{l3}) do not have any real solution for 
		all cases of $x',y',z'$. Hence $M$ does not have any 3-OO decomposition. Further, there is a real solution for the equations (\ref{l5})-(\ref{l7}): $l_1=-0.589222,l_2=0.835328,l_3=0.720831,l_4=-0.876801$. This implies that $M$ has a 4-OO decomposition. The method of obtaining this decomposition is similar to that of Example 2.

		{\bf Example 4:}
		Suppose the singular values of $M$ are $0,1,\sqrt{8}$. By calculation,
		the equations (\ref{l1})-(\ref{l3}) do not have any real solution for 
		all cases of $x',y',z'$. Hence $M$ does not have any 3-OO decomposition.
		Further, the equations (\ref{l5})-(\ref{l7}) also do not have any real solutions. This implies that $M$ does not have any $4$-OO decomposition. Hence $M$ does not have any $n$-OO decomposition.
		\subsection{Two weaker forms of order-three $n$-OO decomposition}
		\label{subsec:weak}
	\begin{lemma}
			\label{le:weak}
			(i) Every  order-three real matrix can be written as the real linear combination of $O_1,O_2,O_3$, where $O_1,O_2,O_3\in \Ort(3)$ and $\tr(O_1^TO_2)=0$. 
			
			(ii) Every order-three real matrix can be written as the real linear combination of of $O_1,O_2,O_3$, where $O_1,O_2\in \Ort(3)$ and $\tr(O_i^TO_j)=0$ for any $i\neq j$. 
		\end{lemma}

		\begin{proof}
		Similar to the proof of Lemma \ref{tran},	we only need to prove that the two claims hold for any positive semidefinite diagonal matrix.
			
			(i) Assume that $M=\bma a&0&0\\0&b&0\\0&0&c \ema$, where $0\le a\le b\le c$ and let 
			\begin{eqnarray}
				N=M+\frac{b-c}{2}\bma 1&0&0\\0&-1&0\\0&0&1 \ema=\bma a+\frac{b-c}{2}&0&0\\0&\frac{b+c}{2}&0\\0&0&\frac{b+c}{2} \ema.
			\end{eqnarray}
			We have $|a+\frac{b-c}{2}|\le a+|\frac{b-c}{2}|\le a+\frac{b+c}{2}\le b+c$. Recall from (\ref{2o}), we obtain that there exists a 2-OO matrix set $\{O_1,O_2\}$ such that $N=k_1O_1+k_2O_2$. Hence $M=k_1O_1+k_2O_2+\frac{(c-b)}{2}\bma 1&0&0\\0&-1&0\\0&0&1 \ema$. So assertion (i) holds.
			
			(ii) The assertion is proved by the following identity:
			\begin{eqnarray}
				\notag
				\bma a&0&0\\0&b&0\\0&0&c \ema=&&\frac{a+b-2c}{6}(\bma \frac{1}{2}&-\frac{\sqrt{3}}{2}&0\\ \frac{\sqrt{3}}{2}&\frac{1}{2}&0\\0&0&-1 \ema+\bma \frac{1}{2}&\frac{\sqrt{3}}{2}&0\\ -\frac{\sqrt{3}}{2}&\frac{1}{2}&0\\0&0&-1 \ema)\\+&&\bma \frac{5a-b+2c}{6}&0&0\\0&\frac{-a+5b+2c}{6}&0\\0&0&\frac{a+b+c}{3} \ema.
			\end{eqnarray}
			One can verify that the first two matrices on the rightside is a 2-OO matrix set and they are both orthogonal to the third matrix.
		\end{proof}
		
		\section{Applications}
		\label{sec:app}
		In this section, we apply our result in quantum information theory. In subsection \ref{RUM}, we give a characterization of UMEB in a real two-qutrit system. In subsection \ref{app2}, we apply the result of $n$-OU decomposition to entanglement theory, showing that a bipartite pure state can be prepared by the superposition of maximally entangled states with local operations on one of the systems.
		\subsection{Unextendible maximally entangled state bases in $\bbR^3\otimes \bbR^3$}
		\label{RUM}
	We begin with the following definitions.
		\begin{definition}
			A set of bipartite pure states $\{\ket{\psi_i}\}_{i=1}^n\in \bbR^d\otimes \bbR^d$ is called an $n$-number RUMEB if and only if 
			
			(i) all states $\ket{\psi_i}$ are maximally entangled,
			
			(ii) $\langle \psi_i| \psi_j \rangle=\delta_{i,j}$,
			
			(iii) if $\ket{\phi}\in \bbR^d \otimes \bbR^d$ satisfies $\langle \phi | \psi_i \rangle=0$ for all $i=1,\cdots,n$, then $\ket{\phi}$ cannot be a maximally entangled state.
		\end{definition}

		\begin{definition}
			Two sets of RMUEB $\{\ket{\psi_i}\}_{i=1}^n$ and $\{\ket{\phi_i}\}_{i=1}^n$ are called orthogonally equivalent, if there exist real orthogonal matrices $U$ and $V$ such that $U\otimes V\ket{\psi_i}=(-1)^{a_i} \ket{\phi_{\pi(i)}}$ for $i=1,\cdots,n$, $a_i\in \{0,1\}$, and $\pi(i)$ is a permutation of $1,\cdots,n$.
		\end{definition}
		
		Note that each bipartite state $\ket{\psi}$ in $\bbR^d\otimes \bbR^d$ can be uniquely associated with an order-$d$ real matrix $M$ through the identity $\ket{\psi}=(I_d\otimes M)\ket{\Omega_d}$ where   $\ket{\Omega_d}=\frac{1}{\sqrt{d}}\sum_{j=0}^{d-1} \ket{jj}$. 
		Hence the following lemma can be proven directly.
		
		\begin{lemma}
			\label{state}
			(i) A set of bipartite pure states $\{\ket{\psi_i}:=(I_d\otimes O_i)\ket{\Omega_d}\}_{i=1}^n$ is an $n$-number RUMEB if and only if 
			
		\quad	(a) $O_1,\cdots,O(n)\in \Ort(d)$,
			
		\quad	(b) $Tr(O_i^TO_j)=d\delta_{i,j}$,
			
		\quad	(c) there does not exist any $O\in \Ort(d)$ such that $\tr(O^TO_i)=0$ for $i=1 ,\cdots,n$.
			
			(ii) Two sets of  RMUEB $\{\ket{\psi_i}:=(I_d\otimes O_i)\ket{\Omega_d} \}_{i=1}^n$ and $\{\ket{\phi_i}:=(I_d\otimes O_i')\ket{\Omega_d}\}_{i=1}^n$ are orthogonally equivalent iff there exist $U,V\in \Ort(d)$ such that $UO_iV=(-1)^{a_i}O_{\pi(i)}'$ for $i=1,\cdots,n$, $a_i\in \{0,1\}$, and $\pi(i)$ is a permutation of $1,\cdots,n$.
		\end{lemma}

		We now apply the previous results to obtain the following theorem.
		
		\begin{theorem}
			\label{ru}
			Let
			\begin{eqnarray}
				\label{fs}
				\notag
			&&\ket{\psi_i}=(I_3\otimes G_i)\ket{\Omega_3}, i=1,2,3,4,\\
			&&\ket{\psi_3'}=(I_3\otimes G_3')\ket{\Omega_3},
				\end{eqnarray}
			where $G_1,G_2,G_3,G_3',G_4\in \Ort(3)$ are defined in Corollary \ref{3onew} and Theorem \ref{unt}. 
			
			(i) The smallest number of a RUMEB in $\bbR^3\otimes \bbR^3$ is three.  Any three-number RUMEB in $\bbR^3\otimes \bbR^3$ is  orthogonally equivalent to $\{\ket{\psi_1},\ket{\psi_2},\ket{\psi_3'}\}$. 
			
			(ii) The largest number of a RUMEB in $\bbR^3\otimes \bbR^3$ is four.  Any four-number RUMEB in $\bbR^3\otimes \bbR^3$ is  orthogonally equivalent to $\{\ket{\psi_1},\ket{\psi_2},\ket{\psi_3},\ket{\psi_4}\}$.
		\end{theorem}
		
		\begin{proof}
			(i) We have shown in Theorem \ref{2ge} that any 2-OO matrix set is orthogonally equivalent to $\{I_3, \O_1\}$, which is extendible.  Thus any 2-OO matrix set is also extendible. By lemma \ref{state} (i), the number of a RUMEB in $\bbR^3\otimes \bbR^3$ is greater than two.
				Further, recall from Corollary \ref{3onew} that any 3-OO matrix set is orthogonally equivalent to $\mathcal{G}_1$ in (\ref{c1'}) or $\mathcal{G}_2$ in (\ref{c2'}). Theorem \ref{unt} implies that only $\mathcal{G}_1$ is unextendible. Hence by Lemma \ref{state} (ii),  we proved the assertion.
			
			(ii) This is directly proved by Theorem \ref{unt} (ii),(iii) and Lemma \ref{state} (ii).
		\end{proof}
We further obtain an interesting property of the four-number RUMEB in a two-qutrit system from the above theorem.
\begin{corollary}
There exists another maximally entangled state in the span of a four-number RUMEB within a two-qutrit system.
\end{corollary}		
\begin{proof}
By theorem \ref{ru} (ii), we only need to prove that the claim holds for the four-number RUMEB $\{\ket{\psi_1},\ket{\psi_2},\ket{\psi_3},\ket{\psi_4}\}$ in (\ref{fs}). 
Let $\ket{\psi_5}:=\frac{1}{2}\sum_{i=1}^4\ket{\psi_i}=
(I_3\otimes G_5)\ket{\Omega_3}$ where $
G_5=\frac{1}{2}\sum_{i=1}^4G_i$.
A calculation shows that 
$
G_5=\bma\frac{1}{4}&\frac{-5-\sqrt{5}}{8}&\frac{5-\sqrt{5}}{8}\\\frac{-5+\sqrt{5}}{8}&\frac{1}{4}&\frac{5+\sqrt{5}}{8}\\\frac{5+\sqrt{5}}{8}&\frac{5-\sqrt{5}}{8}&\frac{1}{4}\ema \in \Ort(3).
$
This implies that $\ket{\psi_5}$ is also a maximally entangled state.
\end{proof}		
		\subsection{Pure bipartite state prepared by maximally entangled states}
		\label{app2}
		\begin{lemma}
			\label{le:bipartite=MESs}
			Any pure state $\ket{\psi}\in \bbC^d\otimes \bbC^d$ can be written as the superposition of at most $d$ maximally entangled states which are mutually orthogonal. 
		\end{lemma}
		
		\begin{proof}
			We write $\ket{\psi}=(I\otimes M)\ket{\Omega_d}$ where $M$ is an order-$d$ complex matrix. Using Lemma \ref{le:nxnUnitary}, there exists an $d$-OU matrix set $\{U_k\}_{k=0}^{d-1}$ such that $M=\sum_{k=0}^{d-1} c_kU_k$. Thus $\ket{\psi}=\sum_{k=0}^{d-1} c_k\ket{\Phi_k}$ where $\ket{\Phi_k}:=(I\otimes U_k)\ket{\Omega_d}$ are mutually-orthogonal maximally entangled states.
		\end{proof}
		
		The result has the following application in entanglement theory. Suppose we are given only maximally entangled states and we are allowed to perform local  unitary gates on one system only. The target is to create an arbitrary bipartite pure state $\ket{\ps}\in\bbC^d\otimes\bbC^d$. 
		Using Lemma \ref{le:bipartite=MESs}, one can show that $\ket{\ps}$ can be produced by the superposition of at most $d$ maximally entangled states under the above-mentioned environment. 
		
\section{conclusions}
		\label{sec:pro}
We have introduced and characterized the sets of $n$-OU and $n$-OO matrices. We showed that $n\le4$ when the real matrices are order-three, and applied the fact to characterize RUMEB in two-qutrit system. An open problem arising from this paper is to find the maximum number $n$ such that an order-$d$ $n$-OO matrix set exists for $d\ge 5$. Further, we proved that any order-$d$ matrix has a $d$-OU decomposition, but there exist real matrices that do not have any $n$-OO decomposition. This leads to another interesting question: does any order $d=4k$ real matrix have a $d$-OO decomposition, where $k$ is a positive integer? We know this holds for some $d$ based on the the existence of real Hadamard matrices so far \cite{horadam2012hadamard}. However, it has not been proved that the real Hadamard matrix exists for any $d=4k$. Hence this question also remains open.

\section*{ACKNOWLEDGMENTS}
Authors were supported by the NNSF of China (Grant No. 11871089).		
		\bibliographystyle{unsrt}
		
		\bibliography{OO}
		
	\end{document}